\documentclass[a4paper,12pt]{article}
\usepackage[pagewise]{lineno}
\usepackage[T1]{fontenc}
\usepackage[latin1]{inputenc}
\usepackage{amsmath}
\usepackage{amssymb}
\usepackage{amsfonts}
\usepackage{euscript}
\usepackage{fancyhdr}
\usepackage{graphicx}
\usepackage{color}
\usepackage{watermark}
\usepackage{fancybox}
\usepackage{colortbl}
\usepackage{empheq}
\usepackage[english]{babel}
\usepackage[english]{minitoc}
\usepackage{hyperref}
\usepackage[ruled,vlined,linesnumbered]{algorithm2e}

\makeindex 
\newtheorem{theorem}{Theorem}[section]
\newtheorem{corollary}[theorem]{Corollary}
\newtheorem{definition}[theorem]{Definition}
\newtheorem{example}[theorem]{Example}
\newtheorem{lemma}[theorem]{Lemma}
\newtheorem{notation}[theorem]{Notation}
\newtheorem{proposition}[theorem]{Proposition}
\newtheorem{remark}[theorem]{Remark}
\newenvironment{proof}[1][Proof]{\noindent\textbf{#1.} }{\ \rule{0.5em}{0.5em}}

\begin{document}

\title{\textbf{Improved Decoding Algorithm of} \\
\textbf{BD-LRPC Codes} }
\author{{\ Hermann Tchatchiem Kamche\thanks{%
Centre for Cybersecurity and Mathematical Cryptology, The University of
Bamenda, Bamenda, Cameroon}} \\
{\small hermann.tchatchiem@gmail.com} }
\maketitle

\begin{abstract}
A Bounded-Degree Low-Rank Parity-Check (BD-LRPC) code is a rank-metric code
that admits a parity-check matrix whose support is generated by a set of powers of an element. This specific structure of the parity-check matrix was employed to enhance the first phase of the decoding algorithm through the expansion of the syndrome support. However, this expansion decreases the probability of
recovering the error support in the second phase of the decoding algorithm. This paper introduces a novel method based on successive intersections to recover the error support. This method offers two key advantages: it increases the probability of successful decoding and
enables the decoding of a greater number of errors.
\end{abstract}


\bigskip

\section{Introduction}

Rank-metric codes are a class of codes whose codewords can be represented as matrices, and the distance between two codewords is defined by the rank of their difference.
Gabidulin codes are the pioneering family of rank-metric codes \cite{Delsarte1978bilinear,Gbidulin1985theory}. Another significant family of rank-metric codes is the Low-Rank Parity-Check (LRPC) codes  \cite{Gaborit2013low}. Introduced in 2013 by Gaborit et al., LRPC codes have found applications in both cryptography \cite{Gaborit2013low,Aragon2019low} and network coding  \cite{Yazbek2017low,El2018efficient}.

Recently, Franch and Li introduced the Bounded-Degree Low-Rank Parity-Check (BD-LRPC) codes \cite{Franch2025bounded}, a subfamily of LRPC codes characterized by a parity-check matrix whose support is generated by a set of powers of an element. Then, they proved that BD-LRPC codes can be used to provide a solution to the problem raised by
Xing and Yuan \cite{Xing2017new}. This problem involves constructing rank-
metric codes with a constant column-to-row ratio that can be efficiently
decoded beyond half the minimum distance. More precisely, let $\mathcal{C}$ be a rank
metric code consisting of matrices of size $m\times n$ over a finite field $%
\mathbb{F}_{q}$. The code rate of $\mathcal{C}$ is $R=\frac{\log _{q}|%
\mathcal{C|}}{mn}$ and the column-to-row ratio is $b=\frac{n}{m}$. Now the
problem is: for a given constant ratio $b=\frac{n}{m}$, explicitly construct
rank-metric codes of rate $R$ with decoding radius $\rho >(1-R)/2$ and
efficiently list decode them. In \cite[Remark 5.]{Franch2025bounded}, Franch
and Li proved that BD-LRPC codes can be used to solve this problem. 

The main advantage of BD-LRPC codes is its decoding algorithm. Recall that in the rank metric, the support of a vector is the subspace generated by its components. The decoding process for LRPC codes is structured into three phases. The first phase involves utilizing the syndrome to identify the product of the error support and the support of the parity-check matrix.
The second phase focuses on recovering the error support, while the third phase is dedicated to
determining the error based on its support. To enhance the first phase of the decoding algorithm, Aragon et al. proposed two expansion functions that transform the syndrome support into another subspace \cite{Aragon2019low}.  In \cite{Franch2025bounded}, Franch and Li exploited the specific structure of the parity check matrix of BD-LRPC codes to define a new expansion function, which significantly improved the success probability of the first decoding phase. However, they only calculated the probability of the first phase for certain specific parameters and made a conjecture on the bounds of this probability in the general case. Furthermore, they did not provide the exact value of the probability of the third phases of the decoding algorithm. Finally, the probability of the second phase is affected by the parameter used in the expansion.

In this paper, we have given a lower bound on the probability of the first phase of the decoding algorithm of BD-LRPC codes, based on the work of \cite{Arora2021unimodular,Semaev2021probabilistic}. The method used to obtain this lower bound allowed us to resolve the conjecture of Franch and Li \cite[Conjecture 1]{Franch2025bounded}. Then, we have demonstrated that the probability of the second phase of the decoding algorithm can be enhanced through the use of carefully designed successive intersections. The primary advantage of this approach is its ability to increase the probability of successful decoding while also enabling the correction of a greater number of errors. Lastly, we have used the work of \cite{Renner2020low} to prove that the probability obtained during the second phase can be used to determine the success probability of the third phase. This connection was used to improve the overall probability of the decoding algorithm based on the existing literature \cite{Franch2025bounded,Burle2023upper}.

The organization of the paper is as follows. In Section \ref{SPreli}, we present essential definitions and notations, along with a description of the new method to recover support via successive intersections. In Section \ref{SDeco}, we describe the decoding algorithm for BD-LRPC codes. In Section \ref{SProb}, we establish a lower bound on the success probability of the decoding algorithm. In Section \ref{ProbaPt}, we give some properties of the probability of the second phase of the decoding algorithm. Finally, we conclude in Section \ref{SCon}.

\section{Preliminaries \label{SPreli}}

\subsection{Rank-Metric Codes}

In this paper, $\mathbb{F}_{q}$ is a finite field with $q$ elements. The set
of all $m\times n$ matrices with entries in $\mathbb{F}_{q}$ will be denoted
by $\mathbb{F}_{q}^{m\times n}$. Let $\mathbf{A}\in \mathbb{F}_{q}^{m\times
n}$, the rank and the transpose of $\mathbf{A}$ are respectively denoted by $rank(\mathbf{A)}$ and $\mathbf{A}^{\intercal }$. The identity matrix of size
$n$ is denoted by $\mathbf{I}_{n}$. Note that $\mathbb{F}_{q}^{m\times n}$
is isomorphic to $\mathbb{F}_{q^{m}}^{n}$. So, to define the rank-metric
codes, we will use the vector representation. 

Let \ $\mathbf{u}=\left( u_{1},\ldots ,u_{n}\right) \in \mathbb{F}%
_{q^{m}}^{n}$. The support of $\mathbf{u}$, denoted by $supp(%
\mathbf{u)}$, is the $\mathbb{F}_{q}$-subspace of $\mathbb{F}_{q^{m}}$
generated by $\left\{ u_{1},\ldots ,u_{n}\right\} $, that is, $supp(\mathbf{%
u)=}\left\langle u_{1},\ldots ,u_{n}\right\rangle _{\mathbb{F}_{q}}$. The
rank of $\mathbf{u}$, denoted by $rank\left( \mathbf{u}\right) $, is the
dimension of the support of $\mathbf{u}$. The rank distance between two
vectors $\mathbf{u}$, $\mathbf{v}$ in $\mathbb{F}_{q^{m}}^{n}$ is the rank
of their differences, that is, $rank\left( \mathbf{u-v}\right) $. A linear code
$\mathcal{C}$ of length $n$ and dimension $k$ is an $\mathbb{F}_{q^{m}}$%
-subspace of $\mathbb{F}_{q^{m}}^{n}$ of dimension $k$. A generator matrix
of $\mathcal{C}$ is an $k\times n$ matrix $\mathbf{G}$ with entries in $%
\mathbb{F}_{q^{m}}$ whose rows generate $\mathcal{C}$. A parity-check matrix
of $\mathcal{C}$ is an $(n-k)\times n$ matrix $\mathbf{H}$ of rank $n-k$
with entries in $\mathbb{F}_{q^{m}}$ such that $\mathbf{GH}^{\intercal }=%
\mathbf{0}$. An element $\mathbf{c}$ in $\mathcal{C}$ is called a codeword.
In the rank metric, when the codeword $\mathbf{c\in }\mathcal{C}$ is
transmitted, the received word is of the form $\mathbf{y=c+e}$, where $%
\mathbf{e}\in \mathbb{F}_{q^{m}}^{n}$ is an error vector of rank $r$. The
syndrome of the received word is $\mathbf{s=yH}^{\intercal }=\mathbf{eH}%
^{\intercal }$. The syndrome decoding problem is to use syndrome $\mathbf{s}$
to recover the error $\mathbf{e}$.

To end this subsection, we will recall that the Gaussian binomial coefficient given by
\begin{equation*}
\left[
\begin{array}{c}
n \\
k
\end{array}
\right] _{q}:=\prod\limits_{i=0}^{k-1}\frac{q^{n}-q^{i}}{q^{k}-q^{i}}
\end{equation*}
is equal to the number of
subspaces of dimension $k$ in a vector space of dimension $n$ over $\mathbb{F}_{q}$. According to \cite{Koetter2008coding}, its value lies between $%
q^{k(n-k)}$ and $4q^{k(n-k)}$.

\subsection{Product of Two Subspaces}

In this subsection, we give some properties of the product of two subspaces.
We will first recall the definition of a specific subset used to define
BD-LRPC codes.

\begin{definition}
\cite{Franch2025bounded} Let $d$ be a positive integer and $\alpha $ an
element of $\mathbb{F}_{q^{m}}$. The bounded-degree subspace generated by $%
\alpha $ of degree $d$, denoted by $\mathcal{V}_{\alpha ,d}$, is the $%
\mathbb{F}_{q}$-subspace of $\mathbb{F}_{q^{m}}$ of dimension $d$ generated by $1,\alpha
,\ldots ,\alpha ^{d-1}$, that is,
\begin{equation*}
\mathcal{V}_{\alpha ,d}\mathbf{=}\left\langle 1,\alpha ,\ldots ,\alpha
^{d-1}\right\rangle _{\mathbb{F}_{q}}\text{.}
\end{equation*}
\end{definition}

In the following, $\alpha $ is an element of $\mathbb{F}_{q^{m}}$ of degree $m$.

\begin{definition}
\cite{Gaborit2013low} Let $\mathcal{E}$ and $\mathcal{W}$ be two $\mathbb{F}%
_{q}$-subspaces of $\mathbb{F}_{q^{m}}$. The product of $\mathcal{E}$ and $%
\mathcal{W}$, denoted by $\mathcal{EW}$, is the $\mathbb{F}_{q}$-subspace of
$\mathbb{F}_{q^{m}}$ generated by $\left\{ ew:e\in \mathcal{E},\ w\in \mathcal{W}\right\} $.
\end{definition}

Let $\mathcal{E}$ and $\mathcal{W}$ be two $\mathbb{F}_{q}$-subspaces of $%
\mathbb{F}_{q^{m}}$. Then, $\dim \left( \mathcal{EW}\right) \leq \dim \left(
\mathcal{E}\right) \dim \left( \mathcal{W}\right) $. In \cite{Gaborit2013low}%
, a lower bound on the probability that $\dim \left( \mathcal{EW}\right)
=\dim \left( \mathcal{E}\right) \dim \left( \mathcal{W}\right) $ was given.
This lower bound has been improved in \cite{Burle2023upper}. Thus, according
to \cite[Proposition 4]{Burle2023upper}, we have the following:

\begin{proposition}
\label{ProductSubspaces}Let $\mathcal{E}$ be drawn uniformly at random from
the set of $\mathbb{F}_{q}$-subspaces of $\mathbb{F}_{q^{m}}$ of dimension $%
r $. Let $\mathcal{W}$ be a fixed $\mathbb{F}_{q}$-subspace of $\mathbb{F}%
_{q^{m}}$ of dimension $w$ such that $rw\leq m$.\ Then, the probability that
$\dim \left( \mathcal{EW}\right) =rw$ is lower bounded by
\begin{equation*}
\Pr \left( \dim \left( \mathcal{EW}\right) =rw\right) \geq 1-\frac{q^{rw}}{%
q^{m}-q^{r-1}}.
\end{equation*}
\end{proposition}

With the same assumptions as in the preview proposition, let $\{b_{1},\ldots
,b_{w}\}$ be a basis of $\mathcal{W}$. When both subspaces $\mathcal{W}$ and
$\mathcal{EW}$ are known, then it is possible in some cases to recover $%
\mathcal{E}$. In \cite{Gaborit2013low}, Gaborit et al. proposed to calculate
$\cap _{i=1}^{w}b_{i}^{-1}\mathcal{EW}$ in order to recover $\mathcal{E}$
with a probability at least $1-rq^{-m+rw(w+1)/2}$. This lower bound has been
improved in \cite{Burle2023upper} by $1-q^{r(2w-1)}/(q^{m}-q^{r-1})$.

When $\mathcal{W}=\mathcal{V}_{\alpha ,w}$, Franch and Li proposed in \cite%
{Franch2025bounded} to calculate $\mathcal{V}_{\alpha ,w}\mathcal{E}\cap
\alpha ^{-(w-1)}\mathcal{V}_{\alpha ,w}\mathcal{E}$ in order to recover $%
\mathcal{E}$ with an estimated probability $1-q^{-m+r(2w-1)}$. In the
following, we will give a method that recovers $\mathcal{E}$ with a
probability at least $1-q^{r(w+1)}/(q^{m}-q^{r-1})$.

\begin{lemma}
\label{SupportRecover1} Let $j$ be a positive integer and $\mathcal{E}$ be a
$\mathbb{F}_{q}$-subspace of $\mathbb{F}_{q^{m}}$ of dimension $r$. If
\bigskip $dim(\mathcal{V}_{\alpha ,j+1}\mathcal{E)}=r(j+1)$ then
\begin{equation*}
(\alpha ^{-1}\mathcal{V}_{\alpha ,j}\mathcal{E)}\cap \mathcal{V}_{\alpha ,j}%
\mathcal{E}=\mathcal{V}_{\alpha ,j-1}\mathcal{E}.
\end{equation*}
\end{lemma}

\begin{proof}
Assume that \bigskip $dim(\mathcal{V}_{\alpha ,j+1}\mathcal{E)}=r(j+1)$.
Then $\mathcal{V}_{\alpha ,j+1}\mathcal{E}$ is a direct sum of $\mathcal{E}%
,\ \alpha \mathcal{E},\ \ldots ,\ \alpha ^{j}\mathcal{E}$, that is,
\begin{equation*}
\mathcal{V}_{\alpha ,j+1}\mathcal{E}=\mathcal{E}\oplus \alpha \mathcal{E}%
\oplus \cdots \oplus \alpha ^{j}\mathcal{E}.
\end{equation*}%
Let%
\begin{equation*}
x\in (\alpha ^{-1}\mathcal{V}_{\alpha ,j}\mathcal{E})\cap \mathcal{V}%
_{\alpha ,j}\mathcal{E}.
\end{equation*}%
Then, there are $x_{1},x_{2}\ldots ,x_{j}$ in $\mathcal{E}$ and $%
x_{1}^{\prime },x_{2}^{\prime }\ldots ,x_{j}^{\prime }$ in $\mathcal{E}$
such that
\begin{equation*}
x=\alpha ^{-1}x_{1}+x_{2}+\cdots +\alpha ^{j-2}x_{j}=x_{1}^{\prime }+\alpha
x_{2}^{\prime }+\cdots +\alpha ^{j-1}x_{j}^{\prime }.
\end{equation*}%
Therefore,
\begin{equation*}
x_{1}+\alpha x_{2}+\cdots +\alpha ^{j-1}x_{j}=\alpha x_{1}^{\prime }+\alpha
^{2}x_{2}^{\prime }+\cdots +\alpha ^{j}x_{j}^{\prime }.
\end{equation*}%
By the direct sum decomposition of $\ \mathcal{V}_{\alpha ,j+1}\mathcal{E}$,
we have $x_{1}=0$ and $\alpha ^{j}x_{j}^{\prime }=0$. Thus,
\begin{equation*}
x=x_{2}+\cdots +\alpha ^{j-2}x_{j}\in \mathcal{V}_{\alpha ,j-1}\mathcal{E}.
\end{equation*}%
As $\mathcal{V}_{\alpha ,j-1}\mathcal{E\subset }(\alpha ^{-1}\mathcal{V}%
_{\alpha ,j}\mathcal{E})\cap \mathcal{V}_{\alpha ,j}\mathcal{E}$, we have $%
(\alpha ^{-1}\mathcal{V}_{\alpha ,j}\mathcal{E})\cap \mathcal{V}_{\alpha ,j}%
\mathcal{E}=\mathcal{V}_{\alpha ,j-1}\mathcal{E}$.
\end{proof}

\begin{theorem}
\label{SupportRecover2} Let $\mathcal{E}$ be drawn uniformly at random from
the set of $\mathbb{F}_{q}$-subspaces of $\mathbb{F}_{q^{m}}$ of dimension $%
r $ and let $w$ be a positive integer such that $r(w+1)\leq m$. Set $%
\mathcal{F}_{w}=\mathcal{V}_{\alpha ,w}\mathcal{E}$ and
\begin{equation*}
\mathcal{F}_{j-1}=(\alpha ^{-1}\mathcal{F}_{j})\cap \mathcal{F}_{j},\ \ \ \
\ for\ \ j=w,\ldots ,2.
\end{equation*}

(a) If \bigskip $dim(\mathcal{V}_{\alpha ,w+1}\mathcal{E)}=r(w+1)$, then $%
\mathcal{F}_{1}=\mathcal{E}$.

(b) The probability that $\mathcal{F}_{1}=\mathcal{E}$ is lower bounded by:
\begin{equation*}
\Pr (\mathcal{F}_{1}=\mathcal{E})\geq 1-\frac{q^{r(w+1)}}{q^{m}-q^{r-1}}.
\end{equation*}
\end{theorem}

\begin{proof}
(a) Assume that \bigskip $dim(\mathcal{V}_{\alpha ,w+1}\mathcal{E)}=r(w+1)$.
Then $dim(\mathcal{V}_{\alpha ,j+1}\mathcal{E)}=r(j+1)$, for $j$ in $%
\{2,\ldots ,w\}$. Therefore, when we successively apply Lemma \ref%
{SupportRecover1} for $j=w,\ldots ,2$, we obtain $\mathcal{F}_{j-1}=\mathcal{%
V}_{\alpha ,j-1}\mathcal{E}$. Thus,$\ \mathcal{F}_{1}=\mathcal{V}_{\alpha ,1}%
\mathcal{E}=\mathcal{E}$.

(b) By the previous result, we have $\Pr (\mathcal{F}_{1}=\mathcal{E})\geq
\Pr ($\bigskip $dim(\mathcal{V}_{\alpha ,w+1}\mathcal{E)}=r(w+1))$. Thus,
from Proposition \ref{ProductSubspaces} the result follows.
\end{proof}

\section{Decoding BD-LRPC Codes \label{SDeco}}

In this section, we will combine the work of \cite%
{Renner2020low,Franch2025bounded} and the iterative intersections described
in Theorem \ref{SupportRecover2} to give a new decoding algorithm for BD-LRPC
codes.

\begin{definition}
\label{defBD-LPRC}Let $d$, $k$, $n$ be positive integers with $0<k<n$. Let $%
\mathbf{H=}\left( h_{i,j}\right) \in \mathbb{F}_{q^{m}}^{(n-k)\times n}$
such that $\left\langle h_{1,1},\ldots ,h_{\left( n-k\right)
,n}\right\rangle _{\mathbb{F}_{q}}=\mathcal{V}_{\alpha ,d}$ and the row vectors of $\mathbf{H}$ are linearly independent. A \textbf{bounded-degree low-rank
parity-check} (BD-LRPC) \textbf{code} with parameters $k$, $n$, $d$, $\alpha
$ is a code with a parity-check matrix $\mathbf{H}$.
\end{definition}

In the following $\mathbf{H}$ is as in Definition \ref{defBD-LPRC}. According to \cite{Renner2020low}, we give the following:

\begin{definition}
\label{ParityCheckMatrixExt}For $i=1,\ldots ,n$ and $j=1,\ldots ,n-k$, let $%
h_{i,j,v}\in \mathbb{F}_{q}$ such that $h_{i,j}=\sum_{v=0}^{d-1}h_{i,j,v}%
\alpha ^{v}$. Set
\begin{equation*}
\mathbf{H}_{ext}=\left(
\begin{array}{cccc}
h_{1,1,0} & h_{1,2,0} & \cdots & h_{1,n,0} \\
h_{1,1,1} & h_{1,2,1} & \cdots & h_{1,n,1} \\
\vdots & \vdots & \ddots & \vdots \\
h_{2,1,0} & h_{2,2,0} & \cdots & h_{2,n,0} \\
h_{2,1,1} & h_{2,2,1} & \cdots & h_{2,n,1} \\
\vdots & \vdots & \ddots & \vdots%
\end{array}%
\right) \in \mathbb{F}_{q}^{\left( n-k\right) d\times n}.
\end{equation*}

Then, $\mathbf{H}$ has:

$\bullet $ the \textbf{unique-decoding property} if $d\geq \frac{n}{n-k}$
and the column vectors of $\mathbf{H}_{ext}$ are linearly independent;

$\bullet $ the \textbf{maximal-row-span property} if every row of $\mathbf{H}
$ spans the entire space $\mathcal{V}_{\alpha ,d}$.
\end{definition}

The following lemma gives a method to recover the error using its support.
\begin{lemma}
\label{ErasureDecoding} Let $\mathcal{E}$ be a $\mathbb{F}_{q}$-subspace of $
\mathbb{F}_{q^{m}}$ of dimension $r$ such that $dim(\mathcal{V}_{\alpha ,d}%
\mathcal{E)}=dr$. Let $\mathbf{s}=\left( s_{1},\ldots ,s_{n-k}\right) \in
\mathbb{F}_{q^{m}}^{n-k}$ such that $s_{i}\in \mathcal{V}_{\alpha ,d}%
\mathcal{E}$, for $i=1,\ldots ,n-k$. Let $\left\{ \varepsilon _{1},\ldots
,\varepsilon _{r}\right\} $ be a basis of $\mathcal{E}$ and $s_{i,u,v}$ in $%
\mathbb{F}_{q}$ such that, for all $i\in \left\{ 1,\ldots ,n-k\right\} $,
\begin{equation}
s_{i}\mathbf{=}\sum_{1\leq u\leq r,0\leq v\leq d-1}s_{i,u,v}\varepsilon
_{u}\alpha ^{v}.  \label{SyndromeDecomposition1}
\end{equation}%
The existence of $\mathbf{e=}\left( e_{1},\ldots ,e_{n}\right) \in
\mathbb{F}_{q^{m}}^{n}$ such that $supp(\mathbf{e)}\subset \mathcal{E}$ and $\mathbf{e}$ is a solution of 
\begin{equation}
\mathbf{He}^{\intercal}=\mathbf{s}^{\intercal}  \label{SyndromeEquation1}
\end{equation}%
with $e_{i,j}$ in $\mathbb{F}_{q}$ such that, for all $i\in \left\{ 1,\ldots
,n\right\} $,
\begin{equation}
e_{i}=\sum\limits_{u=1}^{r}e_{i,u}\varepsilon _{u}.
\label{ErrorDecomposition1}
\end{equation}%
 is equivalent to the following equation
\begin{equation}
\mathbf{H}_{ext}\left(
\begin{array}{c}
e_{1,u} \\
\vdots \\
e_{n,u}%
\end{array}%
\right) =\left(
\begin{array}{c}
s_{1,u,0} \\
s_{1,u,1} \\
\vdots \\
s_{2,u,0} \\
s_{2,u,1} \\
\vdots%
\end{array}%
\right)  \label{ErasureEquation1}
\end{equation}%
$\ $for$\ u=1,\ldots ,r$, with unknowns $e_{i,u}$. Moreover, if $\mathbf{H}$
fulfills the unique-decoding property, then there is at most one $\mathbf{e}%
\in \mathbb{F}_{q^{m}}^{n}$ such that $supp(\mathbf{e)}\subset \mathcal{E}$
and $\mathbf{He}^{\intercal }=\mathbf{s}^{\intercal }$.
\end{lemma}

\begin{proof}
See, \cite[Lemma 4]{Renner2020low}.
\end{proof}

We will now describe a method to recover the error support using the
syndrome. This method is based on the expansion of the syndrome support
introduced in \cite{Franch2025bounded}. Let $\mathbf{e}$ $\in \mathbb{F}%
_{q^{m}}^{n}$ be an error of rank $r$ and $\mathcal{E}$ its support. Let the
syndrome be $\mathbf{s}=\mathbf{eH}^{\intercal }$ and $\mathcal{S}$ its
support. Let $t$ be a positive integer. The expansion of the syndrome
support is the subspace $\mathcal{V}_{\alpha ,t}\mathcal{S}$. As $\mathcal{S}%
\subset \mathcal{V}_{\alpha ,d}\mathcal{E}$, we have $\mathcal{V}_{\alpha ,t}%
\mathcal{S}\subset \mathcal{V}_{\alpha ,d+t-1}\mathcal{E}$. According to
Theorem \ref{SupportRecover2} if $\mathcal{V}_{\alpha ,t}\mathcal{S}=%
\mathcal{V}_{\alpha ,d+t-1}\mathcal{E}$ and $dim(\mathcal{V}_{\alpha ,d+t}%
\mathcal{E})=(d+t)r$ then it is possible to recover $\mathcal{E}$ using $%
\mathcal{S}$. So, Theorem \ref{SupportRecover2} and Lemma \ref%
{ErasureDecoding} allow to give Algorithm \ref{DecodingAlgorithm1}.

{%
\begin{algorithm}[]
\label{DecodingAlgorithm1}
\caption{%
Decoding Algorithm for BD-LRPC Codes}
\DontPrintSemicolon
\KwIn{

$%
\bullet $ Parity-check matrix $\mathbf{H}$ (as in Definition \ref
{defBD-LPRC})

$\bullet $ The positive integer $t$.

$\bullet$ $\mathbf{y}=\mathbf{c}+\mathbf{e}
$, such that $\mathbf{c}$ is in
the BD-LRPC code $\mathcal{C}$

\ \ \ given
by $\mathbf{H}$, and $\mathbf{e} \in \mathbb{F}_{q^{m}}^{n}$.
}
\KwOut{Codeword $\mathbf{c}%
^{\prime }$ of $\mathcal{C}$ or
"decoding failure".
}

$\mathbf{s}%
=\left( s_{1},\ldots ,s_{n-k}\right) \longleftarrow \mathbf{yH}
^{\intercal }$

\If{$\mathbf{s=0}$}{
    \Return{$\mathbf{y}$}
}

$\mathcal{S}
\longleftarrow \left \langle s_{1},\ldots
,s_{n-k}\right \rangle _{\mathbb{F}_{q}}$

$\mathcal{F}_{d+t-1} \longleftarrow \mathcal{V}_{\alpha ,t}\mathcal{S}$

\For{$j=d+t-1,\ldots ,2 $}{
$\mathcal{F}_{j-1}\longleftarrow (\alpha ^{-1}\mathcal{F}_{j})\cap \mathcal{F}_{j}$
}

$\mathcal{E} \longleftarrow \mathcal{F}_{1}$

\If{$\mathcal{E}=\{0\} $}{
    \Return{"decoding failure"}
}%

$r\longleftarrow \dim( \mathcal{E}) $%

Choose a basis $\left\{ \varepsilon
_{1},\ldots ,\varepsilon _{r}\right\}$ of
$\mathcal{E}$.

Solve Equation (\ref{SyndromeDecomposition1}) with unknowns $s_{i,u,v}$.

\eIf{ Equation (\ref{SyndromeDecomposition1}) has no solution}{
\Return{"decoding failure"}}{

Use a solution of Equation (\ref{SyndromeDecomposition1}) to solve
\newline Equation (\ref{ErasureEquation1}) with unknowns $e_{i,u}$.}

\eIf{
Equation (\ref{ErasureEquation1}) has no solution}{
    \Return{"decoding failure"}}{
Use a solution of Equation (\ref{ErasureEquation1}) to compute $\mathbf{e}$
as in (\ref{%
ErrorDecomposition1}).

\Return{$\mathbf{y-e}$}
}

\end{algorithm}}

\begin{theorem}
\label{CorrectnessOfAlgorithm} Assume that $\mathbf{H}$ fulfills the
unique-decoding property. Let $\mathcal{E}$ be the support of the error
vector $\mathbf{e}\in \mathbb{F}_{q^{m}}^{n}$ of rank $r$. Then, Algorithm %
\ref{DecodingAlgorithm1} with input $\mathbf{y=c+e}$ returns $\mathbf{c}$ if
the following two conditions are fulfilled:

(i) $\mathcal{V}_{\alpha ,t}\mathcal{S}=\mathcal{V}_{\alpha ,d+t-1}\mathcal{E%
}$;

(ii) $dim(\mathcal{V}_{\alpha ,d+t}\mathcal{E})=(d+t)r$.
\end{theorem}

\begin{proof}
Assume that the two conditions are fulfilled. Then, \ given $\mathcal{F}_{1}$
as in Algorithm \ref{DecodingAlgorithm1}, we have $\mathcal{F}_{1}=\mathcal{E%
}$ by Theorem \ref{SupportRecover2}. As $dim(\mathcal{V}_{\alpha ,d+t}\mathcal{%
E})=(d+t)r$, then $dim(\mathcal{V}_{\alpha ,d}\mathcal{E})=dr$. Thus, by
Lemma \ref{ErasureDecoding}, the result follows.
\end{proof}

\begin{remark}
In general, the success of the decoding algorithm of LRPC codes is based on
three conditions \cite{Burle2023upper}. Theorem \ref{CorrectnessOfAlgorithm}
simply gives two conditions. This is because the condition used in the
second phase is also used in the third phase. More precisely, if
condition (ii) of Theorem \ref{CorrectnessOfAlgorithm} is true, then we will
simultaneously have success in the second and third phases of the decoding
algorithm. This is an advantage we have by combining the successive
intersection method of Theorem \ref{SupportRecover2} and the work of \cite%
{Renner2020low}. Note that a similar approach was also used in \cite%
{Kamche2024low}.
\end{remark}

\section{Probability of the Decoding Algorithm \label{SProb}}

\subsection{Probability of Successful Expansions}

In this subsection, we will use the work of \cite%
{Renner2020low,Franch2025bounded} to study the probability of the expansion
of the syndrome support. As in \cite{Renner2020low}, we will first transform
the initial syndrome into another syndrome such that the coordinates of this
new syndrome are uniformly distributed.  Let $\mathbf{e}\in \mathbb{F}%
_{q^{m}}^{n}$ be a random error of rank $r$ and $\mathcal{E=}supp(\mathbf{e)}
$. Let $\mathbf{s=eH}^{\intercal }$ be the syndrome and $\mathcal{S}=supp(%
\mathbf{s)}$. Let $t$ be an integer such that $r(d+t-1)\leq m$. Assume that $%
dim(\mathcal{V}_{\alpha ,d+t-1}\mathcal{E})=(d+t-1)r$ and that $\mathbf{H}$
has the maximal-row-span property. Let $\mathbf{\varepsilon
}=(\varepsilon _{1},\ldots ,\varepsilon _{r})$ be a basis of $\mathcal{E}$.
Then $\mathbf{e}=\varepsilon \mathbf{E}$ where $\mathbf{E}$ is an $r\times n$
matrix over $\mathbb{F}_{q}$ of rank $r$. Let $\mathbf{E}^{\prime }$ be a
random $r\times n$ matrix over $\mathbb{F}_{q}$. Set $\mathbf{e}^{\prime }%
\mathbf{=\varepsilon E}^{\prime }$, $\mathbf{s}^{\prime }\mathbf{=e}^{\prime
}\mathbf{H}^{\intercal }$, $\mathcal{E}^{\prime }=supp(\mathbf{e}^{\prime }%
\mathbf{)}$, and $\mathcal{S}^{\prime }\mathcal{=}supp(\mathbf{s}^{\prime }%
\mathbf{)}$.

\begin{proposition}
\label{Reduction1}With the above notations, we have
\begin{equation*}
\Pr (\mathcal{V}_{\alpha ,t}\mathcal{S}^{\prime }=\mathcal{V}_{\alpha ,d+t-1}%
\mathcal{E)\leq }\Pr (\mathcal{V}_{\alpha ,t}\mathcal{S}=\mathcal{V}_{\alpha
,d+t-1}\mathcal{E)}.
\end{equation*}
\end{proposition}

\begin{proof}
We have $\mathcal{E}^{\prime }\subset \mathcal{E}$ with equality if and only
if $rank(\mathbf{E}^{\prime })=r$. We also have $\mathcal{V}_{\alpha ,t}%
\mathcal{S}\subset \mathcal{V}_{\alpha ,d+t-1}\mathcal{E}$ and $\mathcal{V}%
_{\alpha ,t}\mathcal{S}^{\prime }\subset \mathcal{V}_{\alpha ,d+t-1}\mathcal{%
E}^{\prime }\subset \mathcal{V}_{\alpha ,d+t-1}\mathcal{E}$. As $\dim(%
\mathcal{V}_{\alpha ,d+t-1}\mathcal{E})=(d+t-1)r$, then a necessary
condition for $\mathcal{V}_{\alpha ,t}\mathcal{S}^{\prime }=\mathcal{V}%
_{\alpha ,d+t-1}\mathcal{E}$ is that $rank(\mathbf{E}^{\prime })=r$.
Therefore,
\begin{equation*}
\Pr (\mathcal{V}_{\alpha ,t}\mathcal{S}^{\prime }=\mathcal{V}_{\alpha ,d+t-1}%
\mathcal{E)}\leq \Pr (\mathcal{V}_{\alpha ,t}\mathcal{S}^{\prime }=\mathcal{V
}_{\alpha ,d+t-1}\mathcal{E}/rank(\mathbf{E}^{\prime })=r\mathcal{)}.
\end{equation*}%
As specified in the proof of \cite[Theorem 11.]{Renner2020low}, if $rank(\mathbf{E}^{\prime })=r$ then $\mathbf{e}^{\prime }$ can
also be considered as a random error of rank $r$. Hence,
\begin{equation*}
\Pr (\mathcal{V}_{\alpha ,t}\mathcal{S}^{\prime }=\mathcal{V}_{\alpha ,d+t-1}
\mathcal{E}/rank(\mathbf{E}^{\prime })=r\mathcal{)=}\Pr (\mathcal{V}_{\alpha
,t}\mathcal{S}=\mathcal{V}_{\alpha ,d+t-1}\mathcal{E)}.
\end{equation*}%
Thus,%
\begin{equation*}
\Pr (\mathcal{V}_{\alpha ,t}\mathcal{S}^{\prime }=\mathcal{V}_{\alpha ,d+t-1}%
\mathcal{E)\leq }\Pr (\mathcal{V}_{\alpha ,t}\mathcal{S}=\mathcal{V}_{\alpha
,d+t-1}\mathcal{E)}.\newline
\end{equation*}%
\hfill
\end{proof}

We will now study the probability that $\mathcal{V}_{\alpha ,t}\mathcal{S}%
^{\prime }=\mathcal{V}_{\alpha ,d+t-1}\mathcal{E}$. As $\mathbf{e}^{\prime }%
\mathbf{=\varepsilon E}^{\prime }$, then
\begin{equation*}
\mathbf{s}^{\prime \intercal }=\mathbf{X}_{1}\varepsilon ^{\intercal
}+\alpha \mathbf{X}_{2}\varepsilon ^{\intercal }+\cdots +\alpha ^{d-1}%
\mathbf{X}_{d}\varepsilon ^{\intercal }
\end{equation*}%
where $\mathbf{X}_{1},\ldots ,\mathbf{X}_{d}$ are matrices over $\mathbb{F}%
_{q}$ of size $(n-k)\times r$. \ Since $\mathbf{E}^{\prime }$ is a random
matrix and $\mathbf{H}$ has the maximal-row-span property, one can prove as
in \cite[Theorem 11.]{Renner2020low} that the entries of the matrices $%
\mathbf{X}_{1},\ldots ,\mathbf{X}_{d}$ are uniformly distributed on $\mathbb{%
F}_{q}$. As in \cite{Franch2025bounded}, we have
\begin{equation*}
\left(
\begin{array}{c}
\mathbf{s}^{\prime \intercal } \\
\alpha \mathbf{s}^{\prime \intercal } \\
\vdots \\
\alpha ^{t-1}\mathbf{s}^{\prime \intercal }%
\end{array}%
\right) =\mathbf{M}_{t}\left(
\begin{array}{c}
\mathbf{\varepsilon }^{\intercal } \\
\alpha \mathbf{\varepsilon }^{\intercal } \\
\vdots \\
\alpha ^{d+t-2}\mathbf{\varepsilon }^{\intercal }%
\end{array}%
\right)
\end{equation*}%
where%
\begin{equation}
\mathbf{M}_{t}=\left(
\begin{array}{ccccccc}
\mathbf{X}_{1} & \mathbf{X}_{2} & \cdots & \mathbf{X}_{d} & \mathbf{0} &
\cdots & \mathbf{0} \\
\mathbf{0} & \mathbf{X}_{1} & \mathbf{X}_{2} & \cdots & \mathbf{X}_{d} &
\cdots & \mathbf{0} \\
\vdots & \ddots & \ddots & \ddots &  & \ddots & \vdots \\
\mathbf{0} & \cdots & \mathbf{0} & \mathbf{X}_{1} & \mathbf{X}_{2} & \cdots &
\mathbf{X}_{d}%
\end{array}%
\right) .  \label{Mat_t}
\end{equation}

\begin{proposition}
\label{Reduction2}With the above notations,
\begin{equation*}
\Pr (\mathcal{V}_{\alpha ,t}\mathcal{S}^{\prime }=\mathcal{V}_{\alpha ,d+t-1}%
\mathcal{E)=}\Pr (rank(\mathbf{M}_{t})=r(d+t-1)).
\end{equation*}
\end{proposition}

\begin{proof} The matrix $\mathbf{M}_{t}$ is a matrix of size $t(n-k)\times r(d+t-1)$.
The subspace $\mathcal{V}_{\alpha ,t}\mathcal{S}^{\prime }$ is generated by
the entries of the matrices $\mathbf{s}^{\prime },\ \alpha \mathbf{s}%
^{\prime },\ldots ,\alpha ^{t-1}\mathbf{s}^{\prime }$. Therefore, $\mathcal{V%
}_{\alpha ,t}\mathcal{S}^{\prime }=\mathcal{V}_{\alpha ,d+t-1}\mathcal{E}$
if and only if $rank(\mathbf{M}_{t})=r(d+t-1)$.
\end{proof}

\begin{notation}
The probability of the above proposition will be denoted by
\begin{equation}
P_{t}:=\Pr (rank(\mathbf{M}_{t})=r(d+t-1)).  \label{Proba_t}
\end{equation}
\end{notation}

Some properties of the probability $P_t$ are given in Section \ref{ProbaPt}. 

\subsection{Overall Probability}

We will now give the main result on the success probability of the decoding
algorithm.

\begin{theorem}
\label{ProbaFinal}Let $\mathbf{H}$ be as in Definition \ref{defBD-LPRC}. Assume
that $\mathbf{H}$ has the unique-decoding property and the maximal-row-span
property. Let $r$ and $t$ be two positive integers such that $r(d+t-1)\leq
t(n-k)$ and $(d+t)r\leq m$.  Let $\mathbf{e}\in \mathbb{F}%
_{q^{m}}^{n}$ be a random error of rank $r$ and $\mathbf{c}$ be a codeword
of the BD-LRPC code with the parity-check matrix $\mathbf{H}$. Then,
Algorithm \ref{DecodingAlgorithm1} with input $\mathbf{y=c+e}$ returns $%
\mathbf{c}$ with a success probability at least:%
\begin{equation*}
\Pr (success)\geq \left( 1-\frac{q^{r(t+d)}}{q^{m}-q^{r-1}}\right) P_{t}
\end{equation*}%
where $P_{t}$ is defined in (\ref{Proba_t}) and some lower bounds of $P_{t}$
are given in Section \ref{ProbaPt}.
\end{theorem}

\begin{proof}
According to Theorem \ref{CorrectnessOfAlgorithm}, if $\mathcal{V}_{\alpha
,t}\mathcal{S}=\mathcal{V}_{\alpha ,d+t-1}\mathcal{E}$ and $dim(\mathcal{V}%
_{\alpha ,d+t}\mathcal{E})=(d+t)r$ then Algorithm \ref{DecodingAlgorithm1}
returns $\mathbf{c}$. Therefore, the success probability is lower bound by
the probability that $\mathcal{V}_{\alpha ,t}\mathcal{S}=\mathcal{V}_{\alpha
,d+t-1}\mathcal{E}$ and $dim(\mathcal{V}_{\alpha ,d+t}\mathcal{E})=(d+t)r$.
From Proposition \ref{ProductSubspaces}, the probability that $dim(\mathcal{V%
}_{\alpha ,d+t}\mathcal{E})=(d+t)r$ is less than $1-\frac{q^{r(t+d)}}{
q^{m}-q^{r-1}}$. From Propositions \ref{Reduction1} and \ref{Reduction2}, the probability that $\mathcal{V}_{\alpha ,t}\mathcal{S}=\mathcal{V}_{\alpha
,d+t-1}\mathcal{E}$ is less than $P_{t}$. Thus, the result follows.
\end{proof}

\begin{remark}
By Proposition \ref{ProbaIneq}, the probability $P_{t}$ increases with $t$. Thus, the parameter $t$ improves the probability of the first phase of the decoding algorithm. However, it affects the probability of the second phase, which is $1-\frac{q^{r(d+t)}}{q^{m}-q^{r-1}}$. To increase this probability, it is necessary to choose a small value of $t$. In \cite{Franch2025bounded}, Franch and Li observed that when $t=\left\lceil r(d-1)/u\right\rceil +1$, then $P_{t}$ is close to its maximum value, which is $P_{r(d-1)}$, where $u=n-k-r$. Thus, in practice, to maximize the probability of the first and second phases of the decoding algorithm, we can choose $t=\left\lceil r(d-1)/u\right\rceil +1$. Recall that in most cases we have $r(d-1) \leq u$, that is, $rd \leq n-k$. This is, for example, the condition used in classical LRPC codes. When $t=\left\lceil r(d-1)/u\right\rceil +1$ and $r(d-1) \leq u$ then $t=2$. As we will see in Remark \ref{Compa}, in the simple case where $t=2$, BD-LRPC codes have a better decoding probability than LRPC codes as the value of $d$ increases. 
\end{remark}

\subsection{Comparison of Decoding Algorithms}
We will now compare the theoretical upper bound of the failure probability of Algorithm \ref{DecodingAlgorithm1} with that of BD-LRPC codes given in \cite{Franch2025bounded} and that of LRPC codes given in \cite{Gaborit2013low} .

\paragraph{BD-LRPC Codes.}
According to Theorem \ref{ProbaFinal}, an upper bound on the failure
probability of Algorithm \ref{DecodingAlgorithm1} is 
\[D_{New}=1-\left( 1-\frac{q^{r(d+t)}}{q^{m}-q^{r-1}}\right) P_{t}.\] 
In \cite{Franch2025bounded}, the success probability of the first phase of
the decoding algorithm is $P_{t}$, and the probability of the second phase
can be estimated by $1-q^{-m+2(d+t-1)r-r}$. Thus, an upper bound on the
failure probability of the decoding algorithm given by Franch and Li in \cite{Franch2025bounded} can be estimated by
\[D_{FL}=1-\left(1-q^{-m+2(d+t-1)r-r}\right) P_{t}.\]
In general $D_{New}\leq D_{FL}$, because $r(d+t)\leq 2(d+t-1)r-r$. Moreover, Algorithm \ref{DecodingAlgorithm1} can decode errors of rank $r$ up to $\frac{m}{d+t}$, while the decoding algorithm given in \cite{Franch2025bounded} can decode errors of rank $r$ up to $\frac{m}{2(d+t)-3}$. Since $\frac{m}{2(d+t)-3} < \frac{m}{d+t}$ and $D_{New}\leq D_{FL}$, then Algorithm \ref{DecodingAlgorithm1} can decode more errors
with a better probability. 

\paragraph{LRPC codes.} Using the work of \cite{Burle2023upper}, an upper bound on the failure
probability of the decoding algorithm of LRPC codes given by Gaborit et al. in \cite{Gaborit2013low} is 
\[D_{G}=\frac{q^{-(n-k)+dr}}{q-1}+\frac{q^{r(2d-1)}}{%
q^{m}-q^{r-1}}+\frac{q^{rd}}{q^{m}-q^{r-1}}.\]
The difference between the decoding algorithms of LRPC codes and BD-LRPC codes was given in \cite{Franch2025bounded}. We simply note that the main difference is on the first phase of the decoding algorithm. In fact, the decoding algorithm of LRPC codes given in \cite{Gaborit2013low} can decode errors of rank $r$ up to $\frac{n-k}{d}$ and BD-LRPC codes can decode errors of rank $r$ up to $\frac{(n-k)t}{d+t-1}$. Furthermore, as mentioned in Remark \ref{Compa}, when $t \geq 2$, and $d$ increases, the successful  decoding probability is higher for BD-LRPC codes than for LRPC codes.

\subsection{Simulation Results}

We implemented Algorithm \ref{DecodingAlgorithm1} using SageMath \cite%
{Sagemath2023}. The source code is available at \url{https://github.com/Tchatchiem/BD-LRPC_Codes}. We observed
that the theoretical bound given in Theorem \ref{ProbaFinal} approaches the
practical decoding probability.

Figure \ref{FgUpBoundCom} gives the graph of $D_{New}$, $D_{FL}$, $D_{G}$,
for $d=2$, $t=2$, $m=37$, $n=32$, $k=16$, $q=2$ and $0\leq r\leq 10$. We
can observe that $D_{New}$ is approximately equal to $D_{FL}$ for $1\leq
r\leq 5$ and that $D_{New}$ is less than $D_{FL}$ for $6\leq r$ $\leq 10$. \

\begin{center}
\begin{figure}[!h]
\includegraphics[width=15cm]{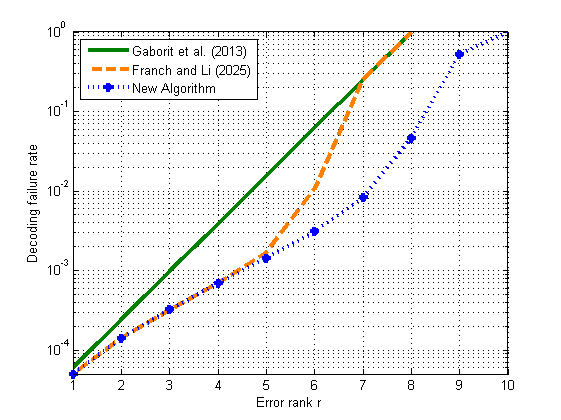}\newline
\caption{Upper bound on the decoding \ failure probability for $d=2$, $t=2$,
$m=37$, $n=32$, $k=16$, $q=2.$}
\label{FgUpBoundCom}
\end{figure}
\end{center}

\section{Some Properties of the Probability \texorpdfstring{$P_{t}$}{Pt}} \label{ProbaPt}

In this section, we present some properties of the probability $P_{t}$,
which corresponds to the probability that the matrix $\mathbf{M}_{t}$,
defined in (\ref{Mat_t}), has rank equal to $r(d+t-1)$. In the following, we set
\begin{equation*}
u:=n-k-r.
\end{equation*}%
This value will be used, as indicated in \cite{Franch2025bounded}, to
calculate $P_{t}$.

\subsection{General results}

Recall that $\mathbf{M}_{t}$ is a matrix of size $t(n-k)\times r(d+t-1)$
defined by
\begin{equation*}
\mathbf{M}_{t}=\left(
\begin{array}{ccccccc}
\mathbf{X}_{1} & \mathbf{X}_{2} & \cdots  & \mathbf{X}_{d} & \mathbf{0} &
\cdots  & \mathbf{0} \\
\mathbf{0} & \mathbf{X}_{1} & \mathbf{X}_{2} & \cdots  & \mathbf{X}_{d} &
\cdots  & \mathbf{0} \\
\vdots  & \ddots  & \ddots  & \ddots  &  & \ddots  & \vdots  \\
\mathbf{0} & \cdots  & \mathbf{0} & \mathbf{X}_{1} & \mathbf{X}_{2} & \cdots
& \mathbf{X}_{d}%
\end{array}%
\right)
\end{equation*}%
where $\mathbf{X}_{i}$ are independently and uniformly chosen from $\mathbb{F%
}_{q}^{(n-k)\times r}$. In \cite[Section 4]{Franch2025bounded} Franch and Li
performed some transformations on the matrix $\mathbf{M}_{t}$ in order to
simplify the calculation of its rank. More precisely, they demonstrated that
the rank of $\mathbf{M}_{t}$ is $r(d+t-1)$ if and only if the rank of $%
\mathbf{X}_{1}$ is $r$ and the rank of the matrix $M_{t}(\mathbf{Z},\mathbf{A%
})$ is $r(d-1)$, where the matrix $M_{t}(\mathbf{Z},\mathbf{A})$ is obtained
from $\mathbf{M}_{t}$ by some elementary transformations and has the
following form:
\begin{equation}
M_{t}(\mathbf{Z},\mathbf{A})=\left(
\begin{array}{c}
\mathbf{Z} \\
\mathbf{ZA} \\
\vdots  \\
\mathbf{ZA}^{t-1}%
\end{array}%
\right)   \label{MatZA}
\end{equation}%
where:

\begin{itemize}
\item $\mathbf{Z}$ is a uniformly random matrix in $\mathbb{F}%
_{q}^{(n-k-r)\times r(d-1)}$;

\item $\mathbf{A}$ is a block matrix of size $r(d-1)\times r(d-1)$ defined
by
\begin{equation}
\mathbf{A=}\left(
\begin{array}{ccccc}
\mathbf{A}_{1} & \mathbf{A}_{2} & \cdots & \mathbf{A}_{d-2} & \mathbf{A}%
_{d-1} \\
\mathbf{I}_{r} & \mathbf{0} & \cdots & \mathbf{0} & \mathbf{0} \\
\mathbf{0} & \mathbf{I}_{r} & \ddots & \vdots & \vdots \\
\vdots & \ddots & \ddots & \mathbf{0} & \mathbf{0} \\
\mathbf{0} & \cdots & \mathbf{0} & \mathbf{I}_{r} & \mathbf{0}%
\end{array}%
\right) \label{Acomp}
\end{equation}
where $\mathbf{A}_{1},\mathbf{A}_{2},\ldots ,\mathbf{A}_{d-1}$ are
independently and uniformly chosen at random from $\mathbb{F}_{q}^{r\times
r} $.
\end{itemize}

Since $\mathbf{X}_{1}$ is a random matrix of \ size $(n-k)\times r$%
, then we have
\begin{equation*}
\Pr (rank(\mathbf{X}_{1})=r)=\prod\limits_{i=0}^{r-1}\left(
1-q^{i-(n-k)}\right) .
\end{equation*}%
Let
\begin{equation}
Q_{t}=\Pr (rank(M_{t}(\mathbf{Z},\mathbf{A}))=r(d-1)). \label{DefQt}
\end{equation}
Then, according to \cite[Proposition 2.]{Franch2025bounded}, we have the
following:

\begin{proposition}
\label{ProbaReduc}The probability $P_{t}$ that $\mathbf{M}_{t}$ has rank
equal to $r(d+t-1)$ is
\begin{equation*}
P_{t}=Q_{t}\prod\limits_{i=0}^{r-1}\left( 1-q^{i-(n-k)}\right) .
\end{equation*}
\end{proposition}

Some properties of the subspace spanned by the rows of the matrix $M_{t}(%
\mathbf{Z},\mathbf{A})$ were presented in \cite[Lemma 4]{Franch2025bounded}.
The direct consequence of this result is as follows.

\begin{corollary}
\label{ProbaIneq} The probability $P_{t}$ has the following properties:
\begin{itemize}
    \item [(i)] $P_{t}\leq $ $P_{t+1}$.
    \item [(ii)]  $P_{t}=$ $P_{r(d-1)}$, for $t\geq r(d-1)$.
\end{itemize}
\end{corollary}

\begin{remark}
\label{ProbaUB}As indicated in \cite{Franch2025bounded}, a necessary
condition for the matrix $\mathbf{M}_{t}$ to have full row rank is that $%
rank(\mathbf{X}_{1})=r$ and $rank(\mathbf{X}_{d})=r$. Therefore, the
probability $P_{t}$ is less than the probability that $rank(\mathbf{X}_{1})=r
$ and $rank(\mathbf{X}_{d})=r$. The matrices $\mathbf{X}_{1}$ and $\mathbf{X}_{d}$ \ are uniformly distributed over $\mathbb{F}_{q}^{(n-k)\times r}$ and the probability that a random $(n-k)\times r$ matrix over $\mathbb{F}_{q}$ has the rank $r$ is equal to $\prod_{i=0}^{r-1}\left(
1-q^{i-(n-k)}\right) $. So, $P_{t}$ is upper bounded by $P_{t}\leq \left(
\prod_{i=0}^{r-1}\left( 1-q^{i-(n-k)}\right) \right) ^{2}.$
\end{remark}

According to Corollary \ref{ProbaIneq}, the optimal value of $P_{t}$ is reached when $t=r(d-1)$. The following theorem shows that this optimal
value of $P_{t}$ is close to the upper bound given in Remark \ref{ProbaUB} .

\begin{theorem}
\label{ProbaExact}Assume that $t=r(d-1)$. Then the probability $P_{t}$ is
equal to
\begin{equation}
P_{t}=\prod\limits_{j=0}^{r-1}\left( 1-q^{j-(n-k)}\right)
\prod\limits_{j=1}^{r}\left( 1-q^{j-(n-k)}\right)  \label{ProbaExact_1}
\end{equation}%
So,
\begin{equation}
1-\frac{q+1}{q-1}q^{-u}\leq P_{t}\leq 1-\frac{q+1}{q}q^{-u}+q^{-2u-1}
\label{ProbaExact_2}
\end{equation}
\end{theorem}

\begin{proof}
 As stated in Proposition \ref{ProbaReduc}, the calculation of $P_{t}$ can be
reduced to the calculation of $Q_{t}$. In Section \ref{ProofProbaExact},
we have shown that 
\begin{equation}
Q_{r(d-1)}=\prod\limits_{j=1}^{r}\left( 1-q^{j-(n-k)}\right) .
\label{ValuQt}
\end{equation} \hfill
\end{proof}

\begin{remark}
The expression of the probability $P_{t}$ given in (\ref
{ProbaExact_1}) was proven in \cite{Franch2025bounded} for the case where $%
d=2$, and for the case where $d\geq 2$ and $r=1$. For the general case,
Franch and Li proposed in \cite[Conjecture 1]{Franch2025bounded} the
following conjecture:
\begin{equation*}
K\leq Q_{r(d-1)}\leq 1-q^{-u}
\end{equation*}%
where
\begin{equation*}
K=\frac{H_{q}((d-1)r+u-1)}{H_{q}(u-1)}
\end{equation*}%
and the function $H_{q}$ is defined by
\begin{equation*}
H_{q}(n)=\prod\limits_{i=1}^{n}(1-q^{-i}).
\end{equation*}
The exact value of $Q_{r(d-1)}$ that we obtained in (\ref{ValuQt}) confirms
this conjecture. Indeed, the upper bound follows from the expression for $Q_{r(d-1)}$. For the lower bound, we have
\begin{equation*}
K=Q_{r(d-1)}\prod\limits_{i=n-k}^{r(d-1)+u-1}(1-q^{-i}).
\end{equation*}%
Thus,
\begin{equation*}
K\leq Q_{r(d-1)}.
\end{equation*}
It should be noted that this inequality has a matrix interpretation. As explained by \cite[Corollary 1]{Franch2025bounded}, the number $K$
represents the probability that the matrix $M_{t}(\mathbf{Z},\mathbf{A})$
has a full row rank when $\mathbf{Z}$ and $\mathbf{A}$ are random matrices
and $t=r(d-1)$. By the definition of $Q_{t}$ in (\ref{DefQt}), $Q_{r(d-1)}$
represents the probability that the matrix $M_{t}(\mathbf{Z},\mathbf{A})$
has a full row rank when $t=r(d-1)$,  $\mathbf{Z}$ is random, and $\mathbf{A}
$ is randomly generated according to (\ref{Acomp}). Thus, when $t=r(d-1)$,
the probability that $M_{t}(\mathbf{Z},\mathbf{A})$ has a full row rank when
$\mathbf{A}$ is random is less than the probability that $M_{t}(\mathbf{Z},%
\mathbf{A})$ has a full row rank when $\mathbf{A}$ is generated as in (\ref%
{Acomp}). By \cite[Lemma 4]{Franch2025bounded}, this result is also true when $r(d-1) \leq t$. It will therefore be interesting to study the comparison of these two probabilities when $t<r(d-1)$. 
\end{remark}

Remember that, the matrix $\mathbf{M}_{t}$ is a matrix of size $t(n-k)\times r(d+t-1)$. Hence, a necessary condition that $rank(\mathbf{M}_{t})=r(d+t-1)$ is that 
$r(d+t-1)\leq t(n-k)$, that is, $\left\lceil r(d-1)/u\right\rceil
\leq t$.
 Thus, in practice, we can choose $t$ such that $\left\lceil r(d-1)/u \right\rceil \leq t\leq (d-1)r$. In \cite{Franch2025bounded},
Franch and Li stated that the experimental results showed that, for a smaller
integer $t\geq \left\lceil r(d-1)/u\right\rceil +1$, $Q_{t}$ is
close to $Q_{(d-1)r}$.  Then, in \cite[ Remark 5.]{Franch2025bounded}, they
provided a lower bound of $P_{t}$ in the case where $d=2$, $t=\left\lceil
r/u \right\rceil +1$, and  $r\geq u$. Note that when $0<r<u$, then $t=\left\lceil
r/u\right\rceil +1=2$. Thus, in Section \ref{ProofProba_d2}, we
calculated $P_{t}$ in the case where $d=2$ and $t=2$.

\begin{proposition}
\label{Proba_d2}\ Assume that $d=2$.

(a) For $r\geq u$ and $t=\left\lceil r/u\right\rceil +1,$%
\begin{equation}
P_{t}\geq 1-\frac{q^{-u/2}+q^{-u}+q^{-u+1}}{1-q}.  \label{Proba_d2_1}
\end{equation}

(b) For $t=2,$
\begin{equation}
P_{2}=\prod\limits_{j=0}^{r-1}\left( 1-q^{j-(n-k)}\right)
\sum_{v=\left\lceil \frac{r}{2}\right\rceil }^{\min
\{r,u\}}q^{-ur}\prod\limits_{i=0}^{v-1}\frac{(q^{u}-q^{i})(q^{r}-q^{i})}{%
q^{v}-q^{i}}\prod\limits_{i=0}^{r-v-1}\left( 1-q^{-v+i}\right) .
\label{Proba_d2_2}
\end{equation}
\end{proposition}

\begin{proof}
The relation (\ref{Proba_d2_1}) is proven in \cite[ Remark 5.]%
{Franch2025bounded}. The relation (\ref{Proba_d2_2}) is proven in Section %
\ref{ProofProba_d2}.
\end{proof}

Since $P_{t}$ increases with $t$, the expression of $P_{2}$ given in (\ref
{Proba_d2_2}) can be used as a lower bound on $P_{t}$ when $d=2$ and $t\geq 2
$. Proposition \ref{Proba_d2} provides a lower bound of $P_{t}$ in the
simple case where $d=2$. The following theorem gives a lower bound in the general case.

\begin{theorem}
\label{ProbaLB}Assume that $\left\lceil r(d-1)/u\right\rceil \leq
t\leq rd$. Then\bigskip\
\begin{equation*}
P_{t}\geq 1-\min_{1\leq j\leq t}\{R_{j}\}
\end{equation*}%
where
\begin{equation*}
R_{j}=q^{r(d-1)-uj}+\sum_{i=1}^{j-1}\left[
\begin{array}{c}
j \\
i%
\end{array}%
\right] _{q}q^{-ui}.
\end{equation*}
\end{theorem}

\begin{proof}
See Section \ref{ProofProbaLB}.
\end{proof}

 In our simulations, we observed that when $t=r(d-1)$ then, in most cases, the value of $j \in \{1,...,t \}$ that minimizes $R_j$  is equal to $\left\lceil r(d-1)/u\right\rceil +1$. These experimental observations are similar to the experimental results of Franch and Li on the probability $Q_t$ \cite{Franch2025bounded}. Thus, in the following corollary, we give a lower bound on $P_t$ for $t=\left\lceil r(d-1)/u\right\rceil +1$.
\begin{corollary}
\label{ProbaLBc} Assume $t=\left\lceil r(d-1)/u\right\rceil +1$.

\begin{description}
\item[(i)] If $t=2$, then
\begin{equation}
P_{t}\geq 1-(q+2)q^{-u}  \label{Proba_LB_1}
\end{equation}

\item[(ii)] If \ $2<t\leq u$, then
\begin{equation}
P_{t}\geq 1-q^{-u}-\frac{1}{q-1}q^{-(u-t)}-\frac{4}{q^{4}-1}q^{-2(u-t)}
\label{Proba_LB_2}
\end{equation}
\end{description}
\end{corollary}

\begin{proof}
See Section \ref{ProofProbaLBc}.
\end{proof}

\begin{remark} \label{Compa}
 When $t\leq u/2$, then the lower bound given in (\ref{Proba_LB_2}) is
better than the lower bound given in (\ref{Proba_d2_1}). The lower bound
given in (\ref{Proba_LB_1}) is close to the one given in (\ref{ProbaExact_2}).
This lower bound corresponds to the case where $u\geq r(d-1)$, that is, $t=2$. Recall that $u=n-k-r$. So, the condition $u\geq r(d-1)$ is equivalent to $%
n-k\geq rd$. When $n-k\geq rd$, the probability of the first phase of the
decoding algorithm of the classical LRPC codes is $
P_{1}=\prod_{i=0}^{rd-1}(1-q^{-(n-k-i)})$, and $P_{1}\geq
1-q^{-(n-k-rd)}/(q-1)$. The lower bound given in (\ref{Proba_LB_1}) does not
depend on $d$. Thus, in the simple case where $t=2$, the probability $P_{2}$
is significantly better than $P_{1}$ as the value of $d$ increases. This
justifies the fact that BD-LRPC codes have a better decoding probability than
LRPC codes.
\end{remark}

\begin{example}
In Table \ref{TablePt}, we present the numerical values of $P_{1}$, $%
P_{r(d-1)}$, and the lower bound $B_{2}=1-\min_{1\leq j\leq 2}\{R_{j}\}$ of $%
P_{2}$, for $q=2$, $n=32$, $k=16$, $d=5$, and $1\leq r\leq 5$. Recall that $%
P_{r(d-1)}$ given in (\ref{ProbaExact_1}) corresponds to the optimal value
of $P_{t}$ when $t$ varies.

\begin{table}[h!]
\centering
\begin{tabular}{|c|c|c|c|}
\hline
$r$ & $P_{1}$ & $B_{2}$ & $P_{r(d-1)}$ \\ \hline
$1$ & $0.99953$ & $0.99991$ & $0.99995$ \\ \hline
$2$ & $0.98447$ & $0.99982$ & $0.99986$ \\ \hline
$3$ & $0.57759$ & $0.99957$ & $0.99968$ \\ \hline
$4$ & $0.00000$ & $0.99536$ & $0.99931$ \\ \hline
$5$ & $0.00000$ & $0.74854$ & $0.99858$ \\ \hline
\end{tabular}
\caption{Numerical Values of $P_{t}$}
\label{TablePt}
\end{table}
\end{example}

\subsection{Proof of Theorem \protect\ref{ProbaExact} \label{ProofProbaExact}%
}

In this subsection, we will give a proof of Theorem \ref{ProbaExact}. According to Proposition \ref{ProbaReduc},
it suffices to show that the product $\prod_{i=1}^{r}(1-q^{n-k-i})$  represents the probability that the matrix
$M_{r(d-1)}(\mathbf{Z},\mathbf{A})$, defined in (\ref{MatZA}), has rank
equal to $r(d-1)$. To calculate this probability, we will use the work of \cite{Arora2021unimodular} on unimodular matrices. As defined in \cite{Arora2021unimodular}, a matrix polynomial $\mathbf{M}$ in $F_{q}[X]^{v\times w}$ is unimodular if
the greatest common divisor of all $s\times s$ minors of $\mathbf{M}$ is
equal to $1$ where $s=\min \{v,w\}$. 

Similarly to the transformation performed in \cite[
Subsec. 4.2.]{Franch2025bounded}, let $\mathbf{B}$ be the anti-diagonal
block matrix for size $r(d-1)\times r(d-1)$ with $\mathbf{I}_{r}$ blocks on
the anti-diagonal, that is,

\begin{equation*}
\mathbf{B=}\left(
\begin{array}{cccc}
\mathbf{0} & \cdots  & \mathbf{0} & \mathbf{I}_{r} \\
\vdots  & \reflectbox{$\ddots$} & \reflectbox{$\ddots$} & \mathbf{0} \\
\mathbf{0} & \mathbf{I}_{r} & \reflectbox{$\ddots$} & \vdots  \\
\mathbf{I}_{r} & \mathbf{0} & \cdots  & \mathbf{0}%
\end{array}
\right)
\end{equation*}
Then, we have
\begin{equation*}
M_{r(d-1)}(\mathbf{Z},\mathbf{A})\mathbf{B}=M_{r(d-1)}(\mathbf{\hat{Z}},%
\mathbf{\hat{A}})
\end{equation*}
where $\mathbf{\hat{Z}}=\mathbf{ZB}$ and
\begin{eqnarray*}
\mathbf{\hat{A}} &=&\mathbf{B}^{-1}\mathbf{AB} \\
&=&\left(
\begin{array}{ccccc}
\mathbf{0} & \mathbf{I}_{r} & \mathbf{0} & \cdots  & \mathbf{0} \\
\mathbf{0} & \mathbf{0} & \mathbf{I}_{r} & \mathbf{\ddots } & \vdots  \\
\vdots  & \vdots  & \mathbf{\ddots } & \mathbf{\ddots } & \mathbf{0} \\
\mathbf{0} & \mathbf{0} & \cdots  & \mathbf{0} & \mathbf{I}_{r} \\
\mathbf{A}_{d-1} & \mathbf{A}_{d-2} & \cdots  & \mathbf{A}_{2} & \mathbf{A}%
_{1}%
\end{array}%
\right)
\end{eqnarray*}%
Thus, the rank of $M_{r(d-1)}(\mathbf{Z},\mathbf{A})$ is $r(d-1)$ if and
only if \ the rank of $M_{r(d-1)}(\mathbf{\hat{Z}},\mathbf{\hat{A}})$ is $%
r(d-1)$. Set
\begin{equation*}
\mathbf{\hat{Z}=}\left(
\begin{array}{ccc}
\mathbf{\hat{Z}}_{1} & \cdots  & \mathbf{\hat{Z}}_{d-1}%
\end{array}%
\right)
\end{equation*}%
where $\mathbf{\hat{Z}}_{j}$ is a submatrix of $\mathbf{\hat{Z}}$ of size $%
(n-k-r)\times r$ \ for $1\leq j\leq d-1$. Since
\begin{equation*}
M_{r(d-1)}(\mathbf{\hat{Z}},\mathbf{\hat{A}})^{\intercal }=\left(
\begin{array}{cccc}
\mathbf{\hat{Z}}^{\intercal } & \mathbf{\hat{A}}^{\intercal }\mathbf{\hat{Z}}%
^{\intercal } & \cdots  & \left( \mathbf{\hat{A}}^{\intercal }\right)
^{r(d-1)}\mathbf{\hat{Z}}^{\intercal }%
\end{array}%
\right)
\end{equation*}%
then, as stated in\ \cite[Remark 2.6]{Arora2021unimodular}, the rank of $%
M_{r(d-1)}(\mathbf{\hat{Z}},\mathbf{\hat{A}})^{\intercal }$ is $r(d-1)$ if
and only if the polynomial matrix $\mathbf{T}$ is unimodular, where:
\begin{eqnarray*}
\mathbf{T} &=&\left(
\begin{array}{cc}
X\mathbf{I}_{r(d-1)}-\mathbf{\hat{A}}^{\intercal } & \mathbf{\hat{Z}}%
^{\intercal }%
\end{array}%
\right)  \\
&=&\left(
\begin{array}{cccccc}
X\mathbf{I}_{r} & \mathbf{0} & \cdots  & \mathbf{0} & -\mathbf{A}%
_{d-1}^{\intercal } & \mathbf{\hat{Z}}_{1}^{\intercal } \\
-\mathbf{I}_{r} & X\mathbf{I}_{r} & \ddots  & \vdots  & -\mathbf{A}%
_{d-2}^{\intercal } & \mathbf{\hat{Z}}_{2}^{\intercal } \\
\mathbf{0} & \ddots  & \ddots  & \mathbf{0} & \vdots  & \vdots  \\
\vdots  & \ddots  & -\mathbf{I}_{r} & X\mathbf{I}_{r} & -\mathbf{A}%
_{2}^{\intercal } & \mathbf{\hat{Z}}_{d-2}^{\intercal } \\
\mathbf{0} & \cdots  & \mathbf{0} & -\mathbf{I}_{r} & X\mathbf{I}_{r}-%
\mathbf{A}_{1}^{\intercal } & \mathbf{\hat{Z}}_{d-1}^{\intercal }%
\end{array}%
\right)
\end{eqnarray*}%
Using the same elementary transformations as in the proof of \cite[Theorem
4.1]{Arora2021unimodular}, we can show that $\mathbf{T}$ is equivalent to
\begin{equation*}
\mathbf{T}^{\prime }\mathbf{=}\left(
\begin{array}{ccccc}
\mathbf{0} & \mathbf{0} & \cdots  & \mathbf{0} & \mathbf{Y} \\
\mathbf{I}_{r} & \mathbf{0} &  &  & \mathbf{0} \\
\mathbf{0} & \ddots  & \ddots  &  & \vdots  \\
\vdots  & \ddots  & \mathbf{I}_{r} & \mathbf{0} & \mathbf{0} \\
\mathbf{0} & \cdots  & \mathbf{0} & \mathbf{I}_{r} & \mathbf{0}%
\end{array}%
\right)
\end{equation*}%
where
\begin{equation*}
\mathbf{Y}=\mathbf{I}_{r,n-k}X^{d-1}+\sum_{i=1}^{d-1}\mathbf{Y}_{i}X^{i-1}
\end{equation*}%
with
\begin{equation*}
\mathbf{Y}_{i}=\left(
\begin{array}{cc}
-\mathbf{A}_{d-i}^{\intercal } & \mathbf{\hat{Z}}_{i}^{\intercal }%
\end{array}%
\right) ,\ \ i=1,\ldots ,d-1
\end{equation*}
and 
\begin{equation*}
\mathbf{I}_{r,n-k}=\left(
\begin{array}{cc}
\mathbf{I}_r & \mathbf{0}
\end{array}%
\right).
\end{equation*}
So, $\mathbf{T}$ is unimodular if and only if $\mathbf{Y}$ is unimodular.
Based on \cite[Theorem 4.1]{Arora2021unimodular}, the probability that $%
\mathbf{Y}$ is unimadular is given by $\prod_{i=1}^{r}(1-q^{i-(n-k)})$. Consequently, \newline $\prod_{i=1}^{r}(1-q^{i-(n-k)})$ represents the
probability that the rank of $M_{r(d-1)}(\mathbf{Z},\mathbf{A})$ is equal to
$r(d-1)$.

To give the bounds on $P_{t}$, it is sufficient to observe that
\begin{equation*}
1-\frac{q^{-u}}{q-1}\leq \ \prod\limits_{j=0}^{r-1}\left(
1-q^{j-(n-k)}\right) \leq 1-q^{-u-1}
\end{equation*}%
and
\begin{equation*}
1-\frac{q^{-u+1}}{q-1}\leq \prod\limits_{j=1}^{r}\left( 1-q^{j-(n-k)}\right)
\leq 1-q^{-u}.
\end{equation*}

\subsection{Proof of Proposition \ref{Proba_d2} (b) \label{ProofProba_d2}}

In this Subsection, we give the proof of the calculation of
the probability $P_{t}$ given in Proposition \ref{Proba_d2} when $d=2$ and $t=2 $. In this case, the matrix $\mathbf{M}_{t}$ defined in (\ref{Mat_t})
becomes

\begin{equation*}
\mathbf{M}_{t}=\left(
\begin{array}{ccc}
\mathbf{X}_{1} & \mathbf{X}_{2} & \mathbf{0} \\
\mathbf{0} & \mathbf{X}_{1} & \mathbf{X}_{2}%
\end{array}%
\right) .
\end{equation*}

To calculate the probability that $rank(\mathbf{M}_{t})=3r$ we will use, as
in \cite{Franch2025bounded}, some elementary transformations on $\mathbf{M}%
_{t}$. Recall that a necessary condition for $rank(\mathbf{M}_{t})=3r$ is
that $rank(\mathbf{X}_{1})=r$ and $rank(\mathbf{X}_{2})=r$.

Assume that $rank(\mathbf{X}_{1})=r$. Then there is an invertible matrix $%
\mathbf{P}$ such that
\begin{equation*}
\mathbf{PX}_{1}=\left(
\begin{array}{c}
\mathbf{I}_{r} \\
\mathbf{0}%
\end{array}
\right) .
\end{equation*}
Set%
\begin{equation*}
\mathbf{P=}\left(
\begin{array}{c}
\mathbf{P}_{1} \\
\mathbf{P}_{2}%
\end{array}%
\right)
\end{equation*}
where $\mathbf{P}_{1}$ and $\mathbf{P}_{2}$ are submatrices of $\mathbf{P}$
of sizes $r\times (n-k)$ and $(n-k-r)\times (n-k)$, respectively. Then, by the elementary transformations of the block matrices, the
matrix $\mathbf{M}_{t}$ is equivalent to the matrix:
\begin{equation*}
\left(
\begin{array}{ccc}
\mathbf{I}_{r} & \mathbf{0} & \mathbf{0} \\
\mathbf{0} & \mathbf{I}_{r} & \mathbf{0} \\
\mathbf{0} & \mathbf{0} & \mathbf{Y}%
\end{array}%
\right)
\end{equation*}
where%
\begin{equation*}
\mathbf{Y=}\left(
\begin{array}{c}
\mathbf{Z} \\
\mathbf{ZA}%
\end{array}%
\right)
\end{equation*}
with $\mathbf{Z=P}_{2}\mathbf{X}_{2}$ and $\mathbf{A=-P}_{1}\mathbf{X}_{2}$.
Therefore, $rank(\mathbf{M}_{t})=3r$ if and only if $rank(\mathbf{Y)}=r$. We
have%
\begin{equation*}
\left(
\begin{array}{c}
-\mathbf{A} \\
\mathbf{Z}%
\end{array}%
\right) =\mathbf{PX}_{2}.
\end{equation*}
Since $\mathbf{P}$ is an invertible matrix and $\mathbf{X}_{2}$ is a random
matrix, then $\mathbf{A}$ and $\mathbf{Z}$ are also random matrices.

Set $v=rank(\mathbf{Z})$. Then, there are invertible matrices $\mathbf{R}$
and $\mathbf{Q}$ such that
\begin{equation*}
\mathbf{RZQ=}\left(
\begin{array}{cc}
\mathbf{I}_{v} & \mathbf{0} \\
\mathbf{0} & \mathbf{0}%
\end{array}%
\right) .
\end{equation*}%
We have
\begin{equation*}
\left(
\begin{array}{cc}
\mathbf{R} & \mathbf{0} \\
\mathbf{0} & \mathbf{R}%
\end{array}%
\right) \mathbf{YQ}=\left(
\begin{array}{c}
\mathbf{RZQ} \\
\mathbf{RZAQ}%
\end{array}%
\right) =\left(
\begin{array}{c}
\mathbf{RZQ} \\
\mathbf{RZQQ}^{-1}\mathbf{AQ}%
\end{array}%
\right) .
\end{equation*}%
Set $\mathbf{A}^{\prime }=\mathbf{Q}^{-1}\mathbf{AQ}$  and
\begin{equation*}
\mathbf{A}^{\prime }=\left(
\begin{array}{cc}
\mathbf{A}_{1}^{\prime } & \mathbf{A}_{2}^{\prime } \\
\mathbf{A}_{3}^{\prime } & \mathbf{A}_{4}^{\prime }%
\end{array}%
\right)
\end{equation*}
where $\mathbf{A}_{1}^{\prime }$, $\mathbf{A}_{2}^{\prime }$, $\mathbf{A}%
_{3}^{\prime }$, $\mathbf{A}_{4}^{\prime }$ are submatrices of $\mathbf{A}%
^{\prime }$ of sizes $v\times v$, $v\times (r-v)$, $(r-v)\times v$, $%
(r-v)\times (r-v)$, respectively. Then, the matrix $\mathbf{Y}$ is
equivalent to the matrix
\begin{equation*}
\left(
\begin{array}{cc}
\mathbf{I}_{v} & \mathbf{0} \\
\mathbf{0} & \mathbf{A}_{2}^{\prime } \\
\mathbf{0} & \mathbf{0}%
\end{array}%
\right) .
\end{equation*}%
Therefore, $rank(\mathbf{Y)}=r$ if and only if $rank(\mathbf{A}_{2}^{\prime
})=r-v$. This implies that$\ r-v\leq v$, that is, $\ \left\lceil \frac{r}{2}%
\right\rceil \leq v$. Since $\mathbf{A}$ is a random matrix and $\mathbf{Q}$
is an invertible matrix, then $\mathbf{A}^{\prime }$ is also a random
matrix. Hence, $\mathbf{A}_{2}^{\prime }$ is a random matrix.

The above transformations show that the probability that $rank(\mathbf{M}%
_{t})=3r$ is equal to the probability that $rank(\mathbf{X}_{1})=r$, $rank(%
\mathbf{Z})=v$, and $rank(\mathbf{A}_{2}^{\prime })=r-v$ with $\left\lceil
\frac{r}{2}\right\rceil \leq v\leq \min \{r,n-k-r\}$.

The probability that $rank(\mathbf{X}_{1})=r$ is equal to the probability
that a random \newline
$(n-k)\times r$ matrix over $\mathbb{F}_{q}$ has rank $r$ which is equal to $\prod\limits_{j=0}^{r-1}\left( 1-q^{j-(n-k)}\right) $.

The probability that $rank(\mathbf{Z})=v$ is equal to the probability that a
random \newline
$(n-k-r)\times r$ matrix over $\mathbb{F}_{q}$ have rank $v$ which is equal
to $q^{-ur}\prod\limits_{i=0}^{v-1}\frac{(q^{u}-q^{i})(q^{r}-q^{i})}{%
q^{v}-q^{i}}$, where $u=n-k-r$.

The probability that $rank(\mathbf{A}_{2}^{\prime })=r-v$ is equal to the
probability that a random $v\times (r-v)$ matrix over $\mathbb{F}_{q}$ has
rank $r-v$ which is equal to $\prod\limits_{i=0}^{r-v-1}\left(
1-q^{-v+i}\right) $.

Thus, the probability that $rank(\mathbf{M}_{t})=3r$ is equal to

\begin{equation*}
\sum\limits_{v=\left\lceil \frac{r}{2}\right\rceil }^{\min
\{r,u\}}q^{-ur}\prod\limits_{j=0}^{r-1}\left( 1-q^{j-(n-k)}\right)
\prod\limits_{i=0}^{v-1}\frac{(q^{u}-q^{i})(q^{r}-q^{i})}{q^{v}-q^{i}}%
\prod\limits_{i=0}^{r-v-1}\left( 1-q^{-v+i}\right) .
\end{equation*}

\subsection{Proof of Theorem \protect\ref{ProbaLB} \label{ProofProbaLB}}

In this subsection, we will give the proof of Theorem \ref{ProbaLB}. This
proof is similar to the one presented in \cite[Section 4]{Semaev2021probabilistic}. Recall that ${M}_{t}$ is an $t(n-k)\times r(d+t-1)$ matrix
defined by
\begin{equation*}
\mathbf{M}_{t}=\left(
\begin{array}{ccccccc}
\mathbf{X}_{1} & \mathbf{X}_{2} & \cdots  & \mathbf{X}_{d} & \mathbf{0} &
\cdots  & \mathbf{0} \\
\mathbf{0} & \mathbf{X}_{1} & \mathbf{X}_{2} & \cdots  & \mathbf{X}_{d} &
\cdots  & \mathbf{0} \\
\vdots  & \ddots  & \ddots  & \ddots  &  & \ddots  & \vdots  \\
\mathbf{0} & \cdots  & \mathbf{0} & \mathbf{X}_{1} & \mathbf{X}_{2} & \cdots
& \mathbf{X}_{d}%
\end{array}%
\right)
\end{equation*}
We have
\begin{equation*}
\Pr (rank(\mathbf{M}_{t})=r(d+t-1))=1-\Pr (rank(\mathbf{M}_{t})<r(d+t-1)).
\end{equation*}%
But, $rank(\mathbf{M}_{t})<r(d+t-1)$ if and only if there is $\mathbf{u}\in
\mathbb{F}_{q}^{r(d+t-1)}\backslash \{\mathbf{0}\}$ such that $\mathbf{M}_{t}%
\mathbf{u}^{\intercal }=\mathbf{0}$. Thus,
\begin{equation*}
\Pr (rank(\mathbf{M}_{t})<r(d+t-1))\leq \sum_{\mathbf{u}\in \mathbb{F}%
_{q}^{r(d+t-1)}\backslash \{\mathbf{0}\}}\Pr (\mathbf{M}_{t}\mathbf{u}%
^{\intercal }=\mathbf{0}).
\end{equation*}%
Let $\mathbf{x}_{i,j}$ be the $i$-th row of the matrix $\mathbf{X}_{j}$, for
$i$ in $\{1,\ldots ,n-k\}$ and $j$ in $\{1,\ldots ,d\}$. Set
\begin{equation*}
\mathbf{M}_{i,t}=\left(
\begin{array}{ccccccc}
\mathbf{x}_{i,1} & \mathbf{x}_{i,2} & \cdots  & \mathbf{x}_{i,d} & \mathbf{0}
& \cdots  & \mathbf{0} \\
\mathbf{0} & \mathbf{x}_{i,1} & \mathbf{x}_{i,2} & \cdots  & \mathbf{x}_{i,d}
& \cdots  & \mathbf{0} \\
\vdots  & \ddots  & \ddots  & \ddots  &  & \ddots  & \vdots  \\
\mathbf{0} & \cdots  & \mathbf{0} & \mathbf{x}_{i,1} & \mathbf{x}_{i,2} &
\cdots  & \mathbf{x}_{i,d}%
\end{array}%
\right) .
\end{equation*}%
for $i$ in $\{1,\ldots ,n-k\}$. Then, $\mathbf{M}_{t}\mathbf{u}^{\intercal }=%
\mathbf{0}$ if and only if $\mathbf{M}_{i,t}\mathbf{u}^{\intercal }=\mathbf{0%
}$ for all $i$ in $\{1,\ldots ,n-k\}$. Set%
\begin{equation*}
\mathbf{u}=(%
\begin{array}{cccc}
\mathbf{u}_{1} & \mathbf{u}_{2} & \cdots  & \mathbf{u}_{d+t-1}),%
\end{array}%
\end{equation*}%
where $\mathbf{u}_{j}\in \mathbb{F}_{q}^{r}$ , for $j$ in $\{1,\ldots
,d+t-1\}$, and
\begin{equation*}
\mathbf{Y}_{\mathbf{u}}=\left(
\begin{array}{cccc}
\mathbf{u}_{1} & \mathbf{u}_{2} & \cdots  & \mathbf{u}_{d} \\
\mathbf{u}_{2} & \mathbf{u}_{3} & \cdots  & \mathbf{u}_{d+1} \\
\vdots  & \vdots  &  & \vdots  \\
\mathbf{u}_{t} & \mathbf{u}_{t+1} & \cdots  & \mathbf{u}_{d+t-1}%
\end{array}%
\right) .
\end{equation*}%
Then,
\begin{equation*}
\mathbf{M}_{i,t}\mathbf{u}^{\intercal }=\mathbf{Y}_{\mathbf{u}}\mathbf{x}%
_{i}^{\intercal }.
\end{equation*}%
Hence, $\mathbf{M}_{i,t}\mathbf{u}^{\intercal }=\mathbf{0}$ if and only if $%
\mathbf{Y}_{\mathbf{u}}\mathbf{x}_{i}^{\intercal }=\mathbf{0}$, which is
equivalent to saying that $\mathbf{x}_{i}^{\intercal }$ is in the right
kernel of $\mathbf{Y}_{\mathbf{u}}$. The probability that $\mathbf{x}%
_{i}^{\intercal }$ is in the right kernel of $\mathbf{Y}_{\mathbf{u}}$ is $%
q^{-rank(\mathbf{Y}_{\mathbf{u}})}$. Since $\mathbf{x}_{1},\ldots ,\mathbf{x}%
_{n-k}$ are independent, we therefore have

\begin{eqnarray*}
\Pr \left( \mathbf{M}_{t}\mathbf{u}^{\intercal }=\mathbf{0}\right)  &=&\Pr
\left( \bigwedge\limits_{i=1}^{n-k}\mathbf{M}_{i,t}\mathbf{u}^{\intercal }=%
\mathbf{0}\right)  \\
&=&\Pr \left( \bigwedge\limits_{i=1}^{n-k}\mathbf{Y}_{\mathbf{u}}\mathbf{x}%
_{i}^{\intercal }=\mathbf{0}\right)  \\
&=&\prod\limits_{i=1}^{n-k}\Pr \left( \mathbf{Y}_{\mathbf{u}}\mathbf{x}%
_{i}^{\intercal }=\mathbf{0}\right)  \\
&=&q^{-(n-k)rank(\mathbf{Y}_{\mathbf{u}})}
\end{eqnarray*}%
Consequently,
\begin{eqnarray*}
\sum_{\mathbf{u}\in \mathbb{F}_{q}^{r(d+t-1)}\backslash \{\mathbf{0}\}}\Pr
\left( \mathbf{M}_{t}\mathbf{u}^{\intercal }=\mathbf{0}\right)  &=&\sum_{%
\mathbf{u}\in \mathbb{F}_{q}^{r(d+t-1)}\backslash \{\mathbf{0}%
\}}q^{-(n-k)rank(\mathbf{Y}_{\mathbf{u}})} \\
&=&\sum_{v=0}^{t-1}N_{v}q^{-(n-k)(t-v)} \\
&&
\end{eqnarray*}%
where $N_{v}$ is the number of $\mathbf{u}\in \mathbb{F}_{q}^{r(d+t-1)}%
\backslash \{\mathbf{0}\}$ such that $rank(\mathbf{Y}_{\mathbf{u}})=t-v$.
Set
\begin{equation*}
S_{v}=\{\mathbf{u}\in \mathbb{F}_{q}^{r(d+t-1)}\backslash \{\mathbf{0}%
\}:rank(\mathbf{Y}_{\mathbf{u}})\leq t-v\}.
\end{equation*}%
Then
\begin{equation*}
N_{v}\leq |S_{v}|
\end{equation*}%
and $\mathbf{u\in }S_{v}$ if and only if there is a subspace $
V\subset \mathbb{F}_{q}^{t}$ of dimension $v$ such that $V$ is contained in
the left kernel of $\mathbf{Y}_{\mathbf{u}}$. When $v=0$ then $%
|S_{v}|=q^{r(d+t-1)}-1$. Assume that $v\neq 0$. Let
\begin{equation*}
\mathbf{B}=\left(
\begin{array}{c}
\mathbf{b}_{1} \\
\vdots  \\
\mathbf{b}_{v}%
\end{array}%
\right) =\left(
\begin{array}{ccc}
b_{1,1} & \cdots  & b_{1,t} \\
&  &  \\
b_{v,1} & \cdots  & b_{v,t}%
\end{array}%
\right)
\end{equation*}%
be a matrix in row echelon form, whose rows span the subspace $V$. As
$V$ is included in the left kernel of $\ \mathbf{Y}_{\mathbf{u}}$, then $%
\mathbf{b}_{i}\mathbf{Y}_{\mathbf{u}}=\mathbf{0}$ , for all $i$ in $\left\{
1,\ldots ,v\right\} $. We have
\begin{equation*}
\mathbf{b}_{i}\mathbf{Y}_{\mathbf{u}}=\mathbf{uA}_{\mathbf{b}_{i}}
\end{equation*}
where
\begin{equation*}
\mathbf{A}_{\mathbf{b}_{i}}=\left(
\begin{array}{ccccc}
b_{i,1}\mathbf{I}_{r} & \mathbf{0} & \mathbf{0} & \cdots  & \mathbf{0} \\
b_{i,2}\mathbf{I}_{r} & b_{i,1}\mathbf{I}_{r} & \mathbf{0} & \cdots  &
\mathbf{0} \\
\vdots  & b_{i,2}\mathbf{I}_{r} & \ddots  & \ddots  &  \\
b_{i,t}\mathbf{I}_{r} & \vdots  & \ddots  & \ddots  & \mathbf{0} \\
~\mathbf{0} & b_{i,t}\mathbf{I}_{r} &  & \ddots  & b_{i,1}\mathbf{I}_{r} \\
\mathbf{0} & \mathbf{0} & \ddots  &  & b_{i,2}\mathbf{I}_{r} \\
\vdots  & \vdots  & \ddots  & \ddots  & \vdots  \\
\mathbf{0} & \mathbf{0} & \cdots  & \mathbf{0} & b_{i,t}\mathbf{I}_{r}%
\end{array}%
\right)
\end{equation*}%
Set
\begin{equation*}
\mathbf{A}_{V}=\left( \mathbf{A}_{\mathbf{b}_{1}}
\begin{array}{cc}
\cdots  & \mathbf{A}_{\mathbf{b}_{v}}
\end{array}
\right) .
\end{equation*}
Then, $\mathbf{uA}_{V}=\mathbf{0}$. Furthermore, $r(v-1)+rd\leq rank(\mathbf{
A}_{V})$. In fact, consider the matrix.
\begin{equation*}
\widetilde{\mathbf{A}}_{V}^{{}}=\left(
\begin{array}{c}
\begin{array}{cccc}
\widetilde{\mathbf{A}}_{\mathbf{b}_{1}}^{{}} & \cdots  & \widetilde{\mathbf{A%
}}_{\mathbf{b}_{v-1}}^{{}} & \mathbf{A}_{\mathbf{b}_{v}}%
\end{array}%
\end{array}%
\right) .
\end{equation*}%
where $\widetilde{\mathbf{A}}_{\mathbf{b}_{j}}^{{}}$ is the matrix formed by
the first $r$ columns of $\mathbf{A}_{\mathbf{b}_{i}}$, for $i=1,...v-1$. \
Since $\mathbf{B}$ is in the row echelon form, the columns of $\widetilde{%
\mathbf{A}}_{V}^{{}}$ are linearly independent. Thus $rank(\widetilde{%
\mathbf{A}}_{V}^{{}})=r(v-1)+rd$ and $rank(\mathbf{A}_{V})\geq rank(%
\widetilde{\mathbf{A}}_{V}^{{}})$. In summary, $\mathbf{u\in }S_{v}$ if and
only if $\mathbf{u}$ is in the left kernel of $\mathbf{A}_{V}$ for some
subspace $V\subset \mathbb{F}_{q}^{t}$ of dimension $v$. Since the
cardinality of the left kernel of $\mathbf{A}_{V}$ is $q^{r\left(
d+t-1\right) -rank(\mathbf{A}_{V})}$, then
\begin{eqnarray*}
|S_{v}| &\leq &\sum_{din(V)=v}^{{}}q^{r\left( d+t-1\right) -rank(\mathbf{A}%
_{V})} \\
&\leq &\sum_{din(V)=v}^{{}}q^{r\left( d+t-1\right) -r(v-1)-rd} \\
&\leq &\left[
\begin{array}{c}
t \\
v%
\end{array}%
\right] _{q}q^{r(t-v)}.
\end{eqnarray*}%
As
\begin{equation*}
\left[
\begin{array}{c}
t \\
v%
\end{array}%
\right] _{q}=\left[
\begin{array}{c}
t \\
t-v%
\end{array}%
\right] _{q},
\end{equation*}%
then we have
\begin{eqnarray*}
\Pr \left( rank(\mathbf{M}_{t})<r(d+t-1)\right)  &\leq
&\sum_{v=0}^{t-1}|S_{v}|q^{-(n-k)(t-v)} \\
&\leq &q^{r(d+t-1)-(n-k)t}+\sum_{v=1}^{t-1}\left[
\begin{array}{c}
t \\
v
\end{array}
\right] _{q}q^{r(t-v)}q^{-(n-k)(t-v)} \\
&\leq &R_{t}
\end{eqnarray*}
where
\begin{equation*}
R_{t}=q^{r(d-1)-ut}+\sum_{v=1}^{t-1}\left[
\begin{array}{c}
t \\
v%
\end{array}%
\right] _{q}q^{-uv}.
\end{equation*}
Therefore,
\begin{equation*}
1-R_{t}\leq P_{t}.
\end{equation*}%
Since $P_{t}$ increases with respect to $t$, then
\begin{equation*}
\max_{1\leq j\leq t}\{1-R_{j}\}\leq P_{t}\text{.}
\end{equation*}

\subsection{Proof of Corollary \protect\ref{ProbaLBc} \label{ProofProbaLBc}}

Recall that

\begin{equation*}
R_{t}=q^{r(d-1)-ut}+\sum_{v=1}^{t-1}\left[
\begin{array}{c}
t \\
v%
\end{array}%
\right] _{q}q^{-uv}.
\end{equation*}%
and
\begin{equation*}
1-R_{t}\leq P_{t}.
\end{equation*}
As $t=\left\lceil r(d-1)/u\right\rceil +1$, then \[t\geq r(d-1)/u+1.\]
Consequently,
\begin{equation*}
q^{r(d-1)-ut}\leq q^{-u}.
\end{equation*}

(i) Assume that $t=2$. Then,
\begin{equation*}
\sum_{v=1}^{t-1}\left[
\begin{array}{c}
t \\
v%
\end{array}%
\right] _{q}q^{-uv}=(q+1)q^{-u}.
\end{equation*}%
Therefore,
\begin{equation*}
P_{2}\geq 1-(q+2)q^{-u}.
\end{equation*}

(ii) Assume that $2<t\leq u$. Then,
\begin{equation*}
\sum_{v=1}^{t-1}\left[
\begin{array}{c}
t \\
v%
\end{array}%
\right] _{q}q^{-uv}=\left[
\begin{array}{c}
t \\
1%
\end{array}%
\right] _{q}q^{-u}+\sum_{v=2}^{t-1}\left[
\begin{array}{c}
t \\
v%
\end{array}%
\right] _{q}q^{-uv}.
\end{equation*}%
As
\begin{equation*}
\left[
\begin{array}{c}
t \\
v
\end{array}%
\right] _{q}\leq 4q^{v(t-v)}
\end{equation*}%
then,
\begin{equation*}
\left[
\begin{array}{c}
t \\
v
\end{array}
\right] _{q}q^{-uv}\leq 4q^{v(t-v-u)}.
\end{equation*}
If $2\leq v$ then \[v(t-v-u)\leq \allowbreak \left( t-u-4\right) v+4.\]
Thus,
\begin{eqnarray*}
\sum_{v=2}^{t-1}\left[
\begin{array}{c}
t \\
v%
\end{array}%
\right] _{q}q^{-uv} &\leq &4\sum_{v=2}^{t-1}q^{\allowbreak \left(
t-u-4\right) v+4} \\
&\leq &4q^{\allowbreak 2\left( t-u-4\right) +4}\frac{1-q^{\left( \allowbreak
t-u-4\right) (t-2)}}{1-q^{\left( t-u-4\right) }} \\
&&
\end{eqnarray*}%
Since $2<t\leq u$, then
\begin{equation*}
1-q^{\left( \allowbreak t-u-4\right) (t-2)}\leq 1
\end{equation*}%
and
\begin{equation*}
\frac{1}{1-q^{\left( t-u-4\right) }}\leq \frac{1}{1-q^{-4}}.
\end{equation*}%
Therefore,%
\begin{equation*}
\sum_{v=2}^{t-1}\left[
\begin{array}{c}
t \\
v%
\end{array}%
\right] _{q}q^{-uv}\leq \frac{4q^{\allowbreak 2\left( t-u\right) }}{q^{4}-1}.
\end{equation*}%
Consequently,
\begin{equation*}
R_{t}\leq q^{-u}+\frac{1}{q-1}q^{-(u-t)}+\frac{4}{q^{4}-1}q^{-2(u-t)}
\end{equation*}%
\begin{equation*}
\end{equation*}

\section{Conclusion \label{SCon}}
\ \ \ \ \  In this paper, we have enhanced the decoding algorithm for Bounded-Degree Low-Rank Parity-Check
(BD-LRPC) codes by introducing a novel strategy in the second phase of the decoding process. This strategy employs a succession of appropriate intersections to effectively recover the error support. We then discuss improvements to the third phase of the decoding algorithm, building on the work of previous researchers \cite{Renner2020low}.

The success probability of the first phase of the decoding algorithm is contingent upon the probability $P_{t}$ that the matrix of the syndrome space expansion achieves full rank. Franch and Li provided a formula to calculate $P_{t}$ when $d = 2$ and $t = r$ \cite{Franch2025bounded}. In our work, we have established a lower bound for $P_{t}$ and presented a formula for the case where $t = r(d-1)$, as well as for the case where $d = 2$ and $t = 2$. The calculation of $P_t$ in the general case where $t = r(d-1)$ allowed us to prove the conjecture of Franch and Li \cite[Conjecture 1]{Franch2025bounded}. 

It is worth mentioning that the rank metric has been defined over finite rings \cite{Kamche2019rank}, and LRPC
codes have been extended to this context \cite{Renner2020low,Renner2021low,Kamwa2021generalization,Kamche2024low}. Given the superior performance of BD-LRPC codes,
exploring the generalization of this family of rank-metric codes over finite rings presents an intriguing
avenue for future research. This expansion could potentially lead to even more robust coding
schemes with enhanced error-correcting capabilities.

\bibliographystyle{IEEEtran}
\bibliography{BoundedLRPC_CodeBib}

\end{document}